\documentclass[leqno]{article}
\usepackage{amsmath}
\usepackage{amsthm}
\usepackage{latexsym}
\usepackage{amssymb}
\usepackage{fullpage}
\usepackage{amsfonts}
\usepackage{graphics}
\usepackage{floatflt}
\newtheorem{theorem}{Theorem}[section]
\newtheorem{remark}{Remark}[section]
\newtheorem{definition}{Definition}[section]
\numberwithin{equation}{section}
\begin{document}
\title{The Kepler Problem: Orbit Cones and Cylinders}
\author{Terry R. McConnell\\
Syracuse University}
\maketitle
\begin{abstract}Planetary orbits, being conic sections, may be obtained as the locus of intersection of planes and cones. The planes involved are familiar to anyone who has studied the classical Kepler problem. We focus here on the cones.
\end{abstract}

\section{Introduction}

Derivations of the orbit equation in the Kepler problem traditionally begin with a reduction to motion in a plane via conservation of angular momentum. This simple, appealing, and effective step is immediate from the rotational symmetry of the force law. What if it were possible similarly to ``read off'' from the form of the equations of motion that solutions must be confined to the surface of certain cones? At the least, this would provide a pedagogically appealing demonstration of the conic section nature of orbits, and, moreover, one that is faithful to the provenance of these curves in solid geometry.

In this paper we study the cones and cylinders of Keplerian orbits. For example, we show in section 4 that orbits of energy $h < 0$ lie on the surface of either of two cylinders with axes that pass through the center of the orbit in the direction of $\bold l \pm \sqrt{\frac{-1}{2mh}} \bold A$, where $\bold l$ is the angular momentum vector, $m$ is the mass, and  $\bold A$ is the Laplace-Runge-Lenz vector.

A conic section uniquely determines the plane in which it lies, but there are infinitely many cones which intersect with a plane to produce a given conic section. It turns out that the locus of vertices of all such cones forms another conic section in a plane perpendicular to the original plane, and with an eccentricity equal to the reciprocal of the eccentricity of the original conic section. We provide complete and self-contained proofs of these rather striking facts in section 2 below. In section 3 we recall the Kepler problem and give rigorous proofs of some qualitative results about orbits. Section 4 is devoted to characterizing the cones and cylinders of planetary orbits in terms of initial conditions, and in section 5 we calculate the location of the cylinder axes for the earth's orbit.

Probably nothing in this paper is new. (How could there be anything new to discover about conic sections or the Kepler problem?) Modern textbooks, however, invariably view conic sections as planar curves, perhaps to the detriment of some very pretty mathematics that only comes into focus when these curves are viewed in a fully 3 dimensional context. We certainly do not claim our approach is pedagogically superior to any of the more standard ones, but hope that it may have some interest.

Throughout the paper we use bold face for vector quantities. Scalar quantities, and, in particular, the magnitude of a vector, are given in normal typeface. Thus, the magnitude of $\mathbf v$ is $v.$ Dots over variables indicate derivatives with respect to time.
\eject

\section{Ellipses in Space}

Let E be an ellipse, considered as a point subset of 3 dimensional Euclidean space. The ellipse uniquely determines the plane, $\Pi,$ that contains it. For study of the plane geometric properties of E it is convenient to adopt as definition of an ellipse the following:

\begin{definition}
An ellipse is the locus of points in the plane, the sum of whose distances from a given pair of distinct points $F_1$ and
$F_2$ (its foci) is equal to a given positive constant.
\end{definition}

(It is customary to include the circle as a degenerate case of the ellipse.)

To state the definition of ellipse as a conic section, it is necessary to clarify some terminology and notation. First, by the word {\it cone} we shall always mean ``right circular cone.'' More precisely, given a line A and a point on it, V, the cone C with {\it vertex} V and {\it axis} A is the locus of points P such that the smaller positive angle from segment $\overline{\text{PV}}$ to line A has a fixed positive value. The lines $\overline{\text{PV}}$ are called {\it generators} of C. The maximum plane angle, $\theta \in (0,\pi),$ between two generators is a characteristic of C that does not appear to have a standard name. We shall take the liberty of coining the term {\it bloom} for the angle $\theta.$

A cone C may naturally be divided into 3 disjoint subsets:
\begin{equation*}
C = C_{-} \cup \{V\} \cup C_{+},
\end{equation*}
where $C_{\pm}$ are congruent connected sets.

If A has an orientation, designate the {\it nappes} $C_{+}$ and $C_{-}$ in the obviously consistent way; otherwise, choose them indifferently.

The definition of ellipse as conic section follows:
\begin{definition}
An ellipse is the locus of intersection of a given cone C and a given plane that is not parallel to any generator of C, provided the plane meets only one nappe of C.
\end{definition}

We shall call the cones C having $E = C \cap \Pi$ the cones {\it belonging to E}. Our goal in this section is to characterize the set of cones belonging to a given ellipse.

For an ellipse, represented as the intersection of a cone C and a plane $\Pi$, a {\it Dandelin sphere} is a sphere centered on the axis A of C, which meets one nappe of C in a circle centered on A, and which meets $\Pi$ in exactly one point. A given ellipse with accompanying cone has exactly two Dandelin spheres lying on opposite sides of $\Pi$.

A standard use of Dandelin spheres (see, e.g., \cite[pp. 320-321]{eves}) is to derive definition 2.1 as a consequence of definition 2.2, yielding along the way the fact that the points of intersection of $\Pi$ with the Dandelin spheres are the foci of E. We use them in a similar way to prove

\begin{theorem} Let E be an ellipse lying in a plane $\Pi$ with foci $F_1$ and $F_2.$ The locus of vertices of cones belonging to E is a subset of the unique plane $\Pi^{\prime}$ through line $\overline{F_1F_2}$ and perpendicular to $\Pi.$
\end{theorem}

\begin{proof}
Let C be a cone belonging to E.
Let $S_1$ and $S_2$ be the two Dandelin spheres, which meet $\Pi$ in points $F_1$ and $F_2$ respectively. Let $C_1$ and $C_2$ be the respective centers of these spheres. Since each sphere is tangent to $\Pi,$ it follows that the lines $\overline{C_1F_1}$ and $\overline{C_2F_2}$ are each perpendicular to $\Pi.$ Therefore these lines are parallel to each other, and since they are distinct, they determine a unique plane $\Pi^{\prime}$. Since both $C_1$ and $C_2$ lie in $\Pi^{\prime}$, it follows that the axis of C is a subset of $\Pi^{\prime}.$ In particular, the vertex of C lies in $\Pi^{\prime}.$ The fact that $\Pi$ and $\Pi^{\prime}$ are perpendicular to each other is immediate from the construction of $\Pi^{\prime}.$
\end{proof}
\eject
\begin{theorem} Let E be an ellipse as in Theorem 2.1. Let segment $\overline{AP}$ be the major axis of E, and determine the plane $\Pi^{\prime}$ as in Theorem 2.1. Then the vertices of cones belonging to E lie on the hyperbola in $\Pi^{\prime}$ having foci A and P and vertices at $F_1$ and $F_2.$
\end{theorem}
\begin{floatingfigure}[l]{2.5in}
\scalebox{.3}{\includegraphics{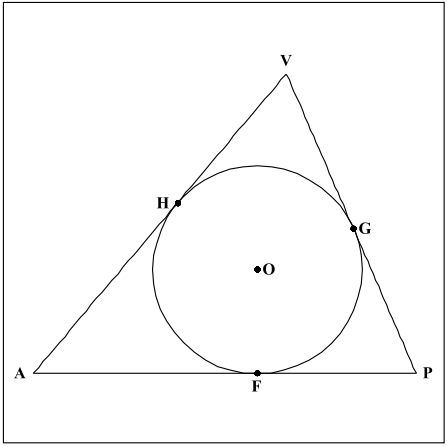}}
\caption{Intersection with plane $\Pi^{\prime}$}
\end{floatingfigure}
{\it Proof.} Let C be a cone belonging to E with vertex V.
Refer to figure 1, in which the plane of the figure corresponds to $\Pi^{\prime}.$ The circle there represents the intersection of $\Pi^{\prime}$ with one of the Dandelin spheres of E.  Point F represents one of the foci of E, and points G and H represent points of contact between the Dandelin sphere and C. Since circle FGH is the incircle of triangle APV, its center O is at the intersection of the bisectors of the internal angles of APV. Denoting by $|\overline{AH}|$ the length of segment $\overline{AH}$, we have
$|\overline{AH}| = |\overline{AF}|, |\overline{PG}| = |\overline{PF}|,$ and  $|\overline{VH}| = |\overline{VG}|.$ Therefore,
\begin{eqnarray*}
|\overline{VA}| - |\overline{VP}| & = &\left(|\overline{VH}| + |\overline{AH}|\right) - \left(|\overline{VG}| + |\overline{PG}|\right)\\ & = & |\overline{AH}| - |\overline{PG}| =  |\overline{AF}| - |\overline{PF}|.
\end{eqnarray*}
The absolute value of the last written expression reduces to $2ea,$ where $e$ is the eccentricity of E, and $a$ is the semi-major axis of E. Since the difference of distances from all such V to the points A and P is constant, it follows that the locus of points V lies on the stated hyperbola. \hfill $\Box$

Denote by $E^{\prime}$ the hyperbola of Theorem 2.2. It has foci at A and P, one vertex at F, and eccentricity $\frac1e.$ The other vertex of $E^{\prime}$ is at the other focus of E. (The excircle of triangle APV, having center at the intersection of $\overline{VO}$ and the bisectors of the exterior angles of APV at A and P, touches $\overline{AP}$ at the other vertex of $E^{\prime}$, and is also the intersection with $\Pi^{\prime}$ of the other Dandelin sphere of E.)
\begin{floatingfigure}[r]{2.5in}
\scalebox{.3}{\includegraphics{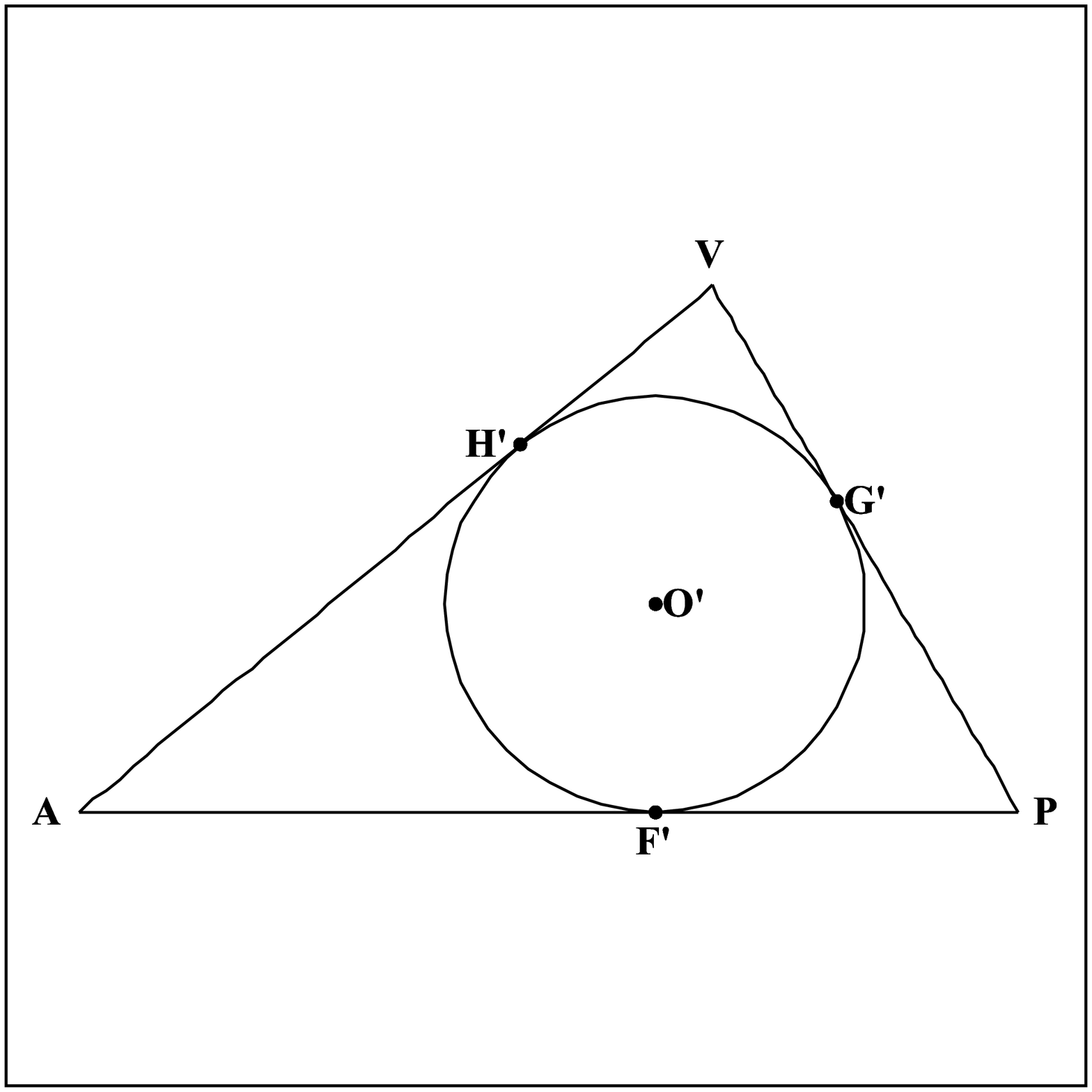}}
\caption{V lying on $E^{\prime}$}
\end{floatingfigure}
Next we show that every point of the hyperbola $E^{\prime},$ except for its vertices, is the vertex of some cone belonging to E. Consider such a point V lying on the branch of $E^{\prime}$ through F. Assume $|\overline{AF}| > |\overline{FP}|$, the contrary case being similar. Construct triangle APV in $\Pi^{\prime}$, and its incircle $H^{\prime}F^{\prime}G^{\prime}$ centered at $O^{\prime}$, as in figure 2.

 By the same reasoning as in the proof of Theorem 2.2, we have

\begin{equation}
 |\overline{AF^{\prime}}| - |\overline{PF^{\prime}}| = |\overline{VA}| - |\overline{VP}| = 2ea.
\end{equation}

In particular, $|\overline{AF^{\prime}}| > |\overline{PF^{\prime}}|.$

Let C be the cone with vertex V, axis $\overline{VO^{\prime}}$, and generators $\overline{VA}$ and $\overline{VP}.$ The intersection of C with $\Pi$ is an ellipse (by definition 2.2), and this ellipse has the sphere with center $O^{\prime}$ and radius
$|\overline{O^{\prime}F^{\prime}}|$ as one of its Dandelin spheres. It also has $\overline{AP}$ as major axis. Let $e^{\prime}$ be the eccentricity of this ellipse. Then it follows from (2.1) that $2ae^{\prime} = 2ae.$
Thus E = $C \cap \Pi,$ and V is the vertex of a cone belonging to E.

The following theorem summarizes much of the foregoing:
\begin{theorem}
Let E be an ellipse with eccentricity $0 < e < 1$ lying in plane $\Pi$ and having major axis $\overline{AP}.$ Let
$\Pi^{\prime}$ be the plane through $\overline{AP}$ perpendicular to $\Pi.$ Then for each $0 < \theta < \pi$ there are exactly 4 cones of bloom $\theta$ belonging to E. Each cone is obtainable from the others by reflection in plane $\Pi$ and in the plane through the center of the ellipse perpendicular to both $\Pi$ and $\Pi^{\prime}.$ The vertices of all such cones lie on the hyperbola in $\Pi^{\prime}$ having foci at A and P and eccentricity $\frac1e$.
\end{theorem}

The axis of a cone belonging to a non-circular ellipse intersects the major axis in a point that lies strictly between the center of the ellipse and the focus that is nearer to the cone vertex. Referring again to figure 1, in which F is the focus closest to V, we have $|\overline{AF}| > |\overline{FP}|,$ hence $|\overline{VA}| > |\overline{VP}|.$ Denoting by J the point of intersection of $\overline{VO}$ and $\overline{AP}$, we have $|\overline{AF}| > |\overline{AJ}| > |\overline{PJ}|.$ To see the left-hand inequality, note that, since the interior angle of APV at P exceeds the angle at A, the interior angle of triangle VJP at J is less than a right angle. On the other hand, $\overline{OF}$ is perpendicular to $\overline{AP}.$ To see the right-hand inequality, apply the law of sines in triangles AJV and VJP to show that the ratio of $|\overline{AJ}|$ to $|\overline{PJ}|$ is equal to the ratio of $|\overline{VA}|$ to $|\overline{VP}|.$

The intersection of a cylinder and a plane not parallel to the cylinder is also an ellipse. We shall say that the cylinder belongs to the ellipse. The following result is the analogue of Theorem 2.3 for cylinders. It may be proved directly or obtained as a limiting case of Theorem 2.3 as the bloom tends to zero.

\begin{theorem}
Let E be an ellipse with eccentricity $0 < e < 1$. Then there are exactly two cylinders belonging to E. The axes of these cylinders pass through the center of the ellipse and form the asymptotes of the hyperbola of loci of vertices of cones belonging to E.
\end{theorem}

For completeness, we conclude by stating the analogous results for circles. If E is a circle lying in plane $\Pi$, then there is exactly one cylinder belonging to E. Its axis is perpendicular to $\Pi$ through the center of E. For each angle $\theta \in (0,\pi)$ there are exactly two cones belonging to E. The vertices of these cones lie on the cylinder axis and are obtained from each other by reflection in $\Pi.$

\section{The Kepler Problem}

According to Newtonian mechanics the position vector $\mathbf r = \mathbf r(t)$ of a body at time $t$ moving under the influence of a Coulomb or gravitational source\footnote{The source is assumed to be rigidly fixed in space. Real two body problems can be reduced to this case.} located at the origin satisfies the equations of motion
\begin{equation}\label{eqmotion}
\ddot{\mathbf r} = -\frac{k\mathbf r}{r^3},
\end{equation}
where $k$ is a measure of the force per unit mass generated by the source. For sentimental reasons we shall call the source ``the sun'' and the body ``the planet.'' Let $\mathbf p = m\dot{\mathbf r}$ be the momentum of the planet, where $m$ is its mass.

As is well known, 4 integrals (conserved quantities) of the motion can readily be deduced from \eqref{eqmotion}: The 3 components of the angular momentum $\mathbf l = \mathbf r \times \mathbf p,$ and the energy, $h = \frac{p^2}{2m} - \frac{mk}{r}.$

We shall be concerned here only with negative energy ($h < 0$) orbits. Such orbits are, as it turns out, generally eccentric ellipses, though there are two annoying special cases: circular orbits and catastrophic ones that plunge straight into the sun. The latter may be avoided by the assumption $\mathbf l \ne 0.$
\theoremstyle{plain}

\begin{theorem}
If $\mathbf l \ne 0$ and $h < 0$ then we have the following {\it a priori} bounds on $r$ and $p$:
\begin{equation}
\frac{l^2}{2m^2k} < r < \frac{mk}{|h|};
\end{equation}
\begin{equation}
\frac{l|h|}{mk} < p < \frac{2m^2k}{l}.
\end{equation}
\end{theorem}
\begin{proof}
Since $\mathbf l \ne 0$, $p$ can never vanish. Thus $\frac{mk}r = \frac{p^2}{2m} - h > -h,$ and the right-hand side of (3.2) follows.

For the left-hand side of (3.2), we introduce the pivotal quantity $\gamma = \mathbf p \cdot \mathbf r,$ the
{\it virial of Clausius.} By the pythagorean theorem,
\begin{equation}
r^2p^2 = \gamma^2 + l^2.
\end{equation}
Since $h < 0$ we have $\frac{mk}r > \frac{p^2}{2m},$ hence $p^2r^2 < 2m^2kr.$ This, in conjunction with (3.4) yields the left-hand side of (3.2).

Again, by (3.4) and (3.2) we have $p^2 \ge \frac{l^2}{r^2} > \frac{l^2h^2}{m^2k^2},$ and the left-hand side of (3.3) follows.

Finally, combining $h < 0$ with the left-hand side of (3.2) we have $\frac{p^2}{2m} < \frac{mk}r \le \frac{2m^3k^2}{l^2}$, and the right-hand side of (3.3) follows.
\end{proof}

Theorem 3.1 and standard results on existence and uniqueness of solutions of ordinary differential equations (see, e.g., \cite[pp. 15-19]{siegel}) show that, given $\mathbf r_0$ and $\mathbf p_0$ satisfying $\mathbf r_0 \times \mathbf p_0 \ne \mathbf 0$ and $\frac{p_0^2}{2m} - \frac{mk}{r_0} < 0,$ the system \eqref{eqmotion} has a unique solution $\mathbf r$ satisfying $\mathbf r(0) = \mathbf r_0$ and $m\dot{\mathbf r}(0) = \mathbf p_0.$ The components of $\mathbf r$ are, moreover, analytic functions of $t \in (-\infty, \infty).$

If $f$ is a function of time, we define its {\it time average,} $\langle f \rangle$ by
\begin{equation*}
\langle f \rangle = \lim_{t \to \infty}\frac1t \int_0^t f(s)\,ds,
\end{equation*}
assuming the limit exists. Also denote by $T = \frac{p^2}{2m}$ and $V = -\frac{mk}r$ the kinetic and potential energy respectively.

Our next result is a very special case of the celebrated {\it Virial Theorem} of statistical mechanics.
\begin{theorem}
Under the same hypotheses as Theorem 3.1, both $\langle T \rangle$ and $\langle V \rangle$ exist and satisfy
\begin{equation}
\langle T \rangle = -\frac12\langle V \rangle = -h.
\end{equation}
\end{theorem}

\begin{proof}
By Theorem 3.1 each of the functions $T,V,$ and $\gamma$ is bounded. Thus, by the fundamental theorem of calculus,
$\langle \dot{\gamma}\rangle = 0.$ On the other hand, $\dot{\gamma} = \frac{p^2}{m}-\frac{mk}r = 2T + V.$ Since
$h = T + V$ obviously has an average value, it follows that both averages $\langle T\rangle$ and $\langle V\rangle$ exist and satisfy (3.5).
\end{proof}
\theoremstyle{remark}

\begin{remark}
The virial $\gamma$ itself is the derivative of a bounded function, $I = \frac12 mr^2$, the {\it moment of inertia.} Thus we also have $\langle \gamma\rangle = 0.$
\end{remark}

The final result of this section will be needed in Section 4.
\begin{theorem}
Assume $\mathbf l \ne \mathbf 0$ and $h < 0.$ Then $\gamma$ vanishes for infinitely many $t$.
\end{theorem}
\begin{proof}
We shall, in fact, show more: that $\dot{r} = 0$ for infinitely many $t$, which suffices, since $\gamma = m\dot{\mathbf r}\cdot\mathbf r = m\dot{r}r.$ Since $r$ is a bounded smooth function, this can only fail to happen if $\lim_{t \to \infty}r(t) = r_{\infty}$ exists. Now there is only one value for this limit that is compatible with Theorem 3.2, i.e.,
\begin{equation}
r_{\infty} = -\frac{mk}{2h}.
\end{equation}
By energy conservation, $p$ then also has a limit, $p_{\infty},$ and by (3.4) and Remark 3.1 we must have
$r_{\infty}p_{\infty} = l.$ Together with (3.6) this implies a relation between $h$ and $l$:
\begin{equation}
h = -\frac{m^3k^2}{2l^2}.
\end{equation}
On the other hand, by (3.4), $pr \ge l,$ so $h = \frac{p^2}{2m} - \frac{mk}r \ge \frac{l^2}{2mr^2} - \frac{mk}r.$ A little calculus shows that the latter quantity has a negative minimum value of $-\frac{m^3k^2}{2l^2},$ achieved if and only if $r = \frac{l^2}{m^2k}.$ Thus, if (3.7) holds, $r$ must be constant, and in that case $\dot{r}$ is identically zero.
\end{proof}

\begin{remark}
From the above proof, it follows that (3.7) is the {\it circularity condition:} An orbit of finite negative energy is circular if and only if (3.7) holds at any time $t$ (hence, at all times $t$.)
\end{remark}

\section{Cones, Cylinders, and Kepler's Laws}

We continue to assume as in section 3 that the initial conditions for equations (3.1) are such that $h < 0$ and
$\mathbf l \ne \mathbf 0.$ If the solution of (3.1) lies for all time on the surface of a certain cone (or cylinder) $C$, we shall call $C$ an {\it orbit cone} (resp. {\it orbit cylinder.})

The main goal of this section is to prove Kepler's first law -- that planetary orbits are ellipses with the sun at one focus -- by the somewhat unusual means of establishing the existence of an orbit cylinder. In so doing, we shall also classify all possible orbit cylinders and cones, and provide a proof of Kepler's third law.

Let
\begin{equation}
\mathbf A = \mathbf p \times \mathbf l - m^2k \frac{\mathbf r}r
\end{equation}
be the {\it Laplace-Runge-Lenz} (LRL) vector.
\begin{theorem}
The vector $\mathbf A$ is an integral of (3.1), i.e., if $\mathbf r$ satisfies (3.1), then $\mathbf A$ is constant in time.
\end{theorem}
See, e.g., Goldstein's text\footnote{Goldstein uses $k$ for the total gravitational force on the planet when $r = 1$, not the force per unit mass. Thus his $k$ corresponds with our $mk$.} \cite{goldsteintext} for a proof. Also see \cite{goldsteinpaper} for an account of the history of this remarkable quantity.

As we show next, the form and character of orbit cylinders and cones may be deduced from Theorem 4.1, but this approach seems unfair in that it is possible to derive the complete orbit equation in a few deft strokes \cite{goldsteintext} from conservation of the LRL vector. Accordingly, after discussing what can be deduced from Theorem 4.1, we present an alternative approach that does not depend on knowing Theorem 4.1.

\begin{theorem}
Assume $h < 0, \mathbf l \ne \mathbf 0,$ and that (3.7) does not hold, i.e., the orbit is not a circle. Then the quantity $|(\mathbf r + \beta^2\mathbf A)\times(\mathbf l + \beta\mathbf A)|$ is constant in time if and only if
$\beta = \pm\sqrt{\frac{-1}{2mh}}.$
\end{theorem}
\begin{proof}
Since $\mathbf A$ and $\mathbf l$ are orthogonal, we have
\begin{equation*}
\begin{aligned}
|(\mathbf r + \beta^2\mathbf A)\times(\mathbf l + \beta\mathbf A)|^2
 &= |(\mathbf r + \beta^2\mathbf A)\times\mathbf l|^2 + \beta^2|\mathbf r \times \mathbf A|^2 \\
 &= |(\mathbf r + \beta^2\mathbf A)\times\mathbf l|^2 + \beta^2|\mathbf r \times (\mathbf p \times \mathbf l)|^2 \\
 &= l^2(r^2 + 2\beta^2\mathbf r \cdot \mathbf A + \beta^4 A^2 + \beta^2\gamma^2),
 \end{aligned}
 \end{equation*}
 where $\gamma = \mathbf r \cdot \mathbf p$ is the virial introduced in section 3.

 Using $\frac{dr^2}{dt} = \frac2m\gamma,$ $\frac{d \gamma^2}{dt} = 2\gamma (\frac{p^2}m - \frac{m k}r)$, and
 $\frac{d}{dt}(\mathbf r \cdot \mathbf A) = -\frac{m k}r\gamma,$ we find that a sufficient condition for
 $|(\mathbf r + \beta^2\mathbf A)\times(\mathbf l + \beta\mathbf A)|$ to be constant is
 \begin{equation*}
 0 = \frac1m - \beta^2\frac{m k}r + \beta^2\left(\frac{p^2}m - \frac{m k}r\right) = \frac1m + 2h\beta^2.
 \end{equation*}

 Conversely, the conclusion $\beta^2 = \frac{-1}{2mh}$ follows unless it happens that $\gamma$ is identically zero. But in that case it is easy to show that (3.7) must hold, a possibility that is explicitly ruled out in the statement of the theorem.
 \end{proof}

 For given $\beta, \mathbf A,$ and $\mathbf l$, the locus of points $\mathbf r$ satisfying
 $|(\mathbf r + \beta^2\mathbf A)\times(\mathbf l + \beta\mathbf A)|$ = constant $> 0$ is a cylinder with axis passing through the point at vector $-\beta^2\mathbf A$ offset from the sun, in the direction of the vector $\mathbf l + \beta\mathbf A.$ Since the orbit is confined to a plane with normal vector $\mathbf l$, it follows from Theorem 4.2 that non-circular orbits lie on the locus of intersection of a cylinder and a plane that is not parallel to the cylinder, i.e., an ellipse. 

 Turning to an analogous result for orbit cones, let $0 < \theta < \pi$ and set $\rho = \tan(\frac{\theta}2).$ Let
 $\mathbf x$ and $\mathbf y$ be fixed non-zero vectors. Then the locus of points $\mathbf r$ satisfying
 \begin{equation}
 \frac{|(\mathbf r - \mathbf x)\times \mathbf y|}{y} = \rho\left| \mathbf y - \frac{(\mathbf r-\mathbf x)\cdot\mathbf y}{y^2}\mathbf y\right|
 \end{equation}
 is a cone of bloom $\theta$ with axis passing through $\mathbf x$ in the direction of $\mathbf y$ and having vertex at $\mathbf x + \mathbf y.$

 Motivated by Theorem 4.2, we seek constants $\alpha \ne 0, \beta \ne 0,$ and $\epsilon > 0$ so that the cone with
 $\mathbf x = \alpha\mathbf A$ and $\mathbf y = \epsilon(\mathbf l + \beta\mathbf A)$ is an orbit cone. With these choices, equation (4.2) reduces to the form $Cr^2 + Dr + E = 0$, where
 $$
 C = l^2(1 + 2mh\beta^2) - m^4k^2\rho^2\beta^2
$$
$$
 D = 2m^2kl^2(\beta^2 + \alpha) - 2m^2k\rho^2\beta(\epsilon l^2 + \epsilon\beta^2A^2 + \alpha\beta A^2 - \beta l^2)
$$
$$
 E = (-2\alpha l^2 + \alpha^2 A^2 - \beta^2 l^2)l^2 - \rho^2(\epsilon l^2 + \epsilon\beta^2A^2 + \alpha\beta A^2 - \beta l^2)^2.
 $$

 (In deriving these these coefficients it is useful to note that, by (3.4), $\gamma^2 = r^2p^2 - l^2 = r^2(\frac{2m^2k}r + 2mh) - l^2.$)

 Setting $C,D,$ and $E$ equal to zero and solving for $\alpha, \beta,$ and $\epsilon$ we find that
 \begin{equation}
  \epsilon = \pm \frac{l}{m^2k\rho}
 \end{equation}
 \begin{equation}
  \beta^2 = \frac{-1}{2mh - \epsilon^{-2}}
 \end{equation}
 \begin{equation}
 \alpha = \frac{1 + \epsilon^{-1}|\beta|}{2mh}
 \end{equation}

 In summary, we have
 \begin{theorem}
 Assume $h < 0, \mathbf l \ne \mathbf 0$ and that (3.7) does not hold. Let $0 < \theta < \pi, \rho = \tan(\frac{\theta}2)$ and $\beta, \alpha$ and $\epsilon$ be given by (4.3)-(4.5) above. Then there are 4 orbit cones: For each of the two choices of $\alpha$ in (4.5) there are two cones of bloom $\theta$ with axis passing through $\alpha\mathbf A.$ One has vertex at $\alpha\mathbf A + |\epsilon|(\mathbf l + |\beta|\mathbf A)$ and the other has vertex at $\alpha\mathbf A - |\epsilon|(\mathbf l - |\beta|\mathbf A).$
 \end{theorem}

 In the rest of this section we attempt to derive the existence of an orbit cylinder from the equations of motion in the most direct and elementary manner possible. Our method is to work in cylindrical coordinates and seek solutions that have constant radial polar coordinate. Unfortunately the proper choice of origin and orientation of coordinate axes for such a coordinate system is not obvious, and will have to emerge in the course of the analysis. Indeed, the origin will not be at the sun, since we know from the results of section 2 that orbit cylinder axes only pass through the sun in the case of circular orbits.

 To simplify matters, it is convenient to make some initial choices. First, we continue to assume, as we have above, that $\mathbf l \ne \mathbf 0, h < 0,$ and that the orbit is not circular. Since $\gamma = \mathbf p\cdot\mathbf r$ vanishes at infinitely many times (by Theorem 3.3), we may and do assume that we have $\gamma = 0$ at time $t = 0$. Let L be the line through this initial position of the planet and the sun. Let $\Pi^{\prime}$ be the unique plane through L and normal to $\mathbf p_0$, the initial momentum of the planet.

 Let $(w, \phi, z)$ be the coordinates of the planet in a cylindrical coordinate system whose origin lies on L, whose z axis lies in plane $\Pi^{\prime}$, and which is oriented so that the initial cylindrical angular ($\phi$) coordinate of the planet is zero. (For the corresponding cartesian coordinate system, $\Pi^{\prime}$ is the x-z plane.) We use $w$ for the cylindrical radial coordinate rather than the customary $r$ since we wish to reserve $\mathbf r$ for the sun-planet vector. We may write the cartesian coordinate vector of the sun as $\mathbf S = (S_x, 0, S_z).$ (At this point, the orientation of the x and z coordinate axes in $\Pi^{\prime}$ is arbitrary. The only further requirement we shall impose for now is that the choices be made in such a way that the x-coordinate of $\mathbf r_0$, the initial position of the planet relative to the sun, and $\dot\phi(0)$ are both positive\footnote{It may seem strange to make the choice of coordinate system part of the problem, but this is not uncommon in mechanics. See, for example, the solution of the Kepler problem in \cite{goldsteintext} using the Hamilton-Jacobi method.}.)

 Let $u = S_x\cos \phi.$ Because of the choices made so far, there is a relation connecting $w, z, u,$ and $r:$

 \begin{equation}
 r^2 = w^2 + z^2 -2wu - 2S_zz + S^2.
 \end{equation}

 The equations of motion (3.1), expressed in an equivalent cylindrical coordinate form, read as follows:

 \begin{equation}
 \ddot w - w\dot \phi^2 + \frac{k}{r^3}(w - u) = 0,
 \end{equation}
 \begin{equation}
 2\dot w\dot \phi + w\ddot\phi + \frac{k}{r^3}S_x\sin\phi = 0,
 \end{equation}
 \begin{equation}
 \ddot z + \frac{k}{r^3}(z - S_z) = 0.
 \end{equation}

 We are interested in solutions of (4.7)-(4.9) having constant $w = w_0 > |S_x|,$ since then the set $w = w_0$ is an orbit cylinder. The gambit of locating the origin away from the sun is compensated by the simplifications in (4.7) and (4.8) that result from $\ddot w = \dot w = 0:$

 \begin{equation}
 w\dot\phi^2 = \frac{k}{r^3}(w - u),
 \end{equation}
 \begin{equation}
 w\ddot\phi = -\frac{k}{r^3}S_x\sin\phi.
 \end{equation}

 Indeed, equations (4.10) and (4.11) are readily integrated.  Dividing (4.11) by (4.10) and rearranging, we have
 \begin{equation*}
 \frac{d}{dt}\ln(\dot\phi) = -\frac{d}{dt}\ln(w - u),
 \end{equation*}
 hence
 \begin{equation}
 \dot\phi = \frac{c}{w - u}, c > 0.
 \end{equation}

 Equation (4.12), in turn, is equivalent to {\it Kepler's Equation:}

 \begin{theorem} A $C^1$ function $\phi$ satisfies (4.12) with $\phi(0) = 0$ if and only if
 \begin{equation}
 \phi(t) - \left(\frac{S_x}{w}\right)\sin\phi(t) = \frac{c}{w}t
 \end{equation}
 holds for all $-\infty < t < \infty.$
 \end{theorem}

 Substituting (4.12) in (4.10), we find that (4.10) and (4.11) entail that $r$ must be an affine function of $u$:

 \begin{equation}
 r = \lambda(w - u), \ \lambda = \left(\frac{k}{c^2w}\right)^{\frac13}.
 \end{equation}

 Conversely, if $\phi = \phi(t)$ satisfies (4.13) with $\phi(0) = 0$, and we put $r = r(t) = \lambda(w - S_x\cos\phi(t)),$ then it is easy to check that (4.10) and (4.11) hold.

 We have yet to consider the z equation, (4.9). Towards this end, note that (4.10) and (4.11) imply
 \begin{equation}
 \ddot u = \frac{k}{r^3}(\frac{S_x^2}w - u).
 \end{equation}

 Next, by completing the square in (4.6) one has the following purely algebraic result.

 \begin{theorem} If (4.14) and (4.6) hold, then $z$ is affine in $u$ if and only if
 $\lambda = (1 - \frac{S_x^2}{w^2})^{-\frac12}.$ If so, then
 \begin{equation}
 z - S_z = \pm\lambda(\frac{S_x^2}{w} - u).
 \end{equation}
 \end{theorem}

 Equation (4.9) is an immediate consequence of (4.15) and (4.16).

 As is well known, initial values $\mathbf r_0$ and $\mathbf p_0$ for the position and momentum of the planet relative to the sun determine a unique solution of the equations of motion. Thus, assuming the parameters $w$, $S_x$, and the constant of integration $c$ in (4.12) can be chosen in a way that is consistent with the initial conditions and the requirement on $\lambda$ from Theorem 4.5, the solution must be exactly that one given by (4.13),(4.14), and (4.16).

 Four of the degrees of freedom in the vector values of $\mathbf r_0$ and $\mathbf p_0$ have already been used to fix the line L and plane $\Pi^{\prime}.$ As we show next, the remaining two degrees of freedom, determined by the magnitudes $p_0$ and $r_0$, together with the algebraic requirement on $\lambda$ in Theorem 4.5, do indeed determine one and only one choice of the parameters.

 There are two cases depending on whether $l^2 > m^2kr_0$ or $l^2 < m^2kr_0$, the case of equality being ruled out by the assumption that the orbit is not a circle.  It turns out that the planet is initially at perihelion in the former case, and initially at aphelion in the latter case. We consider only the first case in detail, since the two are very similar.

 The constraints on $w, S_x,$ and $c$ imposed by the initial conditions are summarized in the following three equations:

 \begin{equation}
 r_0 = k^{\frac13}\left(\frac{w}{c}\right)^{\frac23}\left(1 - \frac{S_x}{w}\right),
 \end{equation}
 \begin{equation}
 l = mk^{\frac13}(w^2c)^{\frac13},
 \end{equation}
 and
 \begin{equation}
 \left(\frac{k}{c^2w}\right)^{\frac13} = \left(1 - \frac{S_x^2}{w^2}\right)^{-\frac12}.
 \end{equation}

 Equation (4.17) is a consequence of (4.14). To obtain (4.18), use $l = r_0p_0 = mr_0w\dot\phi(0),$ (4.12), and (4.17).

 To solve these equations, it is convenient to let $e = \frac{S_x}{w}$, since this will turn out to be the eccentricity of the orbit. First multiply equations (4.18) and (4.19), obtaining
 \begin{equation}
 \left(\frac{w}c\right)^{\frac13} = \frac{l}{mk^{\frac23}\sqrt{1 - e^2}}.
 \end{equation}
 After substituting this result in (4.17), we find
 \begin{equation}
 1 + e = \frac{l^2}{m^2kr_0},
 \end{equation}
 which determines the value of the quantity $e$ from the initial conditions. (We have that $e > 0$, hence $S_x > 0,$ since the right hand side of (4.21) is greater than one - see the paragraph before (4.17). It then follows from (3.2) that $e$ satisfies $ 0 < e < 1.$) Equation (4.19) can be solved for $c$ using (4.20), and finally $w$ can be found from (4.20), (4.22), and (4.17.). The results are:
 \begin{equation}
 c = \frac{mk(1 - e^2)}{l},
 \end{equation}

 \begin{equation}
 w = \frac{(1 - e^2)^{\frac12}r_0}{1 - e},
 \end{equation}
 and
 \begin{equation}
 S_x = we = er_0\sqrt{\frac{1+e}{1-e}}.
 \end{equation}
 (In case 2, one takes $e = -\frac{S_x}{w}$ and proceeds similarly.)

 It follows from similar triangles that $\frac{w}{S_x} = \frac{r_0 + S}{S}$, from which we conclude that
 \begin{equation}
 S = r_0\frac{e}{1-e}.
 \end{equation}

 The location of the origin is now completely determined. It is on the line L at distance $S$ from the sun in the direction opposite the planet. To complete the description of the coordinate system, choose the positive z direction so that the coordinate system is right-handed. Then the z-component of $\mathbf l$ is positive, which implies that the angle, $\psi$, from the sunward side of L to the positive z-axis satisfies $  0 < \psi < \pi.$ Since one has $\sin\psi = \frac{w}{r_0 + S},$ it is easy to derive from (4.23) and (4.25) the following rather pleasing result:
 \begin{equation}
 e = |\cos(\psi)|.
 \end{equation}
The two solutions for $\psi$ correspond to two orbit cylinders. If $\psi$ is acute then both $S_z$ and the z-component of $\mathbf r_0$ are positive. If $\psi$ is obtuse, then both are negative. In either case, another application of similar triangles yields
\begin{equation}
 |S_z| = r_0\frac{e^2}{1 - e}.
 \end{equation}

 Letting $a = \frac{r_0}{1 - e},$ we may write (4.14), the {\it orbit equation}, as
 \begin{equation}
 r = a(1 - e\cos\phi),
 \end{equation}
 where $\phi$ is determined by solving Kepler's equation (4.13). This form reveals that the angular cylindrical coordinate $\phi$ may be interpreted as the {\it eccentric anomaly}, and that the radial cylindrical coordinate, $w$, is equal to $b$, the semi-minor axis of the orbit.

 If $0 < \psi < \frac{\pi}2$ then $S_z > 0.$ If we take the minus sign solution from (4.16) we obtain
 \begin{equation}
 z = ae\cos\phi.
 \end{equation}
 (We can eliminate the positive sign solution because, by (4.25) and similar triangles, we have $z(0) = ae.$)
 If $\frac{\pi}2 < \psi < \pi$ then $S_z < 0.$ In this case, taking the positive sign solution from (4.16) yields $z = -ae\cos\phi.$

It is worth noting the character of the motion in various 2 dimensional projections other than the orbital plane. The pair $(x,y)$ moves around the circular cross section of the orbit cylinder: $(x,y) = (w\cos\phi, w\sin\phi).$ The motion of the pair $(x,z)$ is degenerate, being confined to the line L through the origin of slope $\frac{e}{\sqrt{1-e^2}},$ while the pair (z,y) moves on an ellipse centered at the origin with axes of lengths $ea$ and $\sqrt{1-e^2}a.$

We conclude this section with a derivation of Kepler's 3rd law. Let $\tau$ be the period of the orbit. Then, by (4.13) we have $\tau = 2\pi\frac{w}c.$ Thus, by (4.20) and (4.21),
\begin{equation*}
\tau = \frac{2\pi l a}{mk\sqrt{1-e^2}}.
\end{equation*}
To eliminate $l$, use Kepler's 2nd law (an easy consequence of conservation of momentum) to find $l\tau = 2\pi abm,$ thus obtaining Kepler's 3rd law:
\begin{equation*}
\tau^2 = \frac{4\pi^2}{k}a^3.
\end{equation*}

\section{The Orbit of the Earth}

The orbit cylinder axes, or equivalently, the asymptotes of the set of orbit cone vertices, may be regarded as elements of an orbit. It is therefore of interest to locate these axes for significant orbits such as the orbit of the earth.

The orbit of the earth is not a fixed ellipse, so it is necessary to begin by clarifying what we mean by ``orbit of the earth.'' We shall regard the ``earth'' as an idealized point representing the center of mass of the earth-moon system. Then, following the usual procedure for reducing a two body problem to a one body problem, we regard this point as being in orbit around the center of mass of the solar system, rather than the center of mass of the sun. The resulting orbit is still not an ellipse due to the influence of the other planets. Nevertheless, during a reasonably short time interval -- a year, say -- the orbit is an ellipse to very high precision. Accordingly, we shall define the elliptical orbit of the earth as the best approximating ellipse during a selected year. It is natural to select for this role the year 2000, since modern tables of stellar and planetary positions are referenced to the ecliptic and vernal equinox as they existed at 12 hours coordinated universal time on Jan 1, 2000 (Julian date 2451545.0, or J2000 for short.) The Jet Propulsion Laboratory web site \cite{JPL} provides the following values for selected elements of this orbit: semi-major axis $a = 1.00000261$AU,  eccentricity $e = 0.01671123,$ inclination of orbital plane to reference plane $ I = -0^{\circ}.00001531,$ and longitude of the perihelion $\pi = 102^{\circ}.93768193.$ (Recall that the longitude is measured eastward along the ecliptic from the J2000 vernal equinox.)

The nonzero value of the inclination element deserves a comment. Tabulated celestial coordinates reference the {\it mean} ecliptic and equator at the epoch J2000, not the actual ecliptic and equator. The mean in question is an average over periodic perturbations ( e.g., nutation of the earth's axis of rotation,) which are more significant for the equator and equinox than for the ecliptic. Since we are, in any case, considering an idealized orbit, we shall simply assume the orbit lies in the mean ecliptic and take $I = 0.0.$

The orbit cylinder axes intersect the celestial sphere at 4 points which may be regarded as lying infinitely far from the observer. Thus, the geometric and apparent positions of these points coincide. (Precise calculation of apparent positions of real celestial bodies -- planets, moons, and even stars -- must address the location of the observer, effects due to motion of the observer, and the finite speed of light.) We shall identify the axes by their points of intersection with the celestial sphere.

Since the cylinder axes lie in a plane perpendicular to the ecliptic through the perihelion, the circle of celestial longitude
($\lambda$) of all 4 points is the same as for the perihelion. The celestial latitude, $\beta$, is the same as the angle we denoted $\psi$ in (4.26) above, and is determined by the eccentricity alone. The celestial coordinates of the 4 points are then as follows, where we list the pairs $(\beta,\lambda),$ and in brackets the corresponding declination and right ascension: $(89^{\circ}.0425, 102^{\circ}.9377)[+67^{\circ}.4930, 17\text{h}57\text{m}45.6\text{s}]$,\newline
$(89^{\circ}.0475, 282^{\circ}.9377)[+65^{\circ}.6270, 18\text{h}2\text{m}4.7\text{s}]$,$(-89^{\circ}.0425, 282^{\circ}.9377)[-67^{\circ}.4930, 5\text{h}57\text{m}45.6\text{s}]$,\newline and $(-89^{\circ}.0425, 102^{\circ}.9377)[-65^{\circ}.6270, 6\text{h}2\text{m}4.7\text{s}]$.

\end{document}